\documentclass[twoside]{article}
\usepackage{times,latexsym}
\usepackage{url}
\usepackage[T1]{fontenc}

\usepackage{amsmath,amssymb,amsthm,bbm}
\usepackage[dvipsnames]{xcolor}
\usepackage{tikz-qtree}
\usepackage{endnotes}
\usepackage{graphicx}
\usepackage{endnotes}
\usepackage{hyperref}
\usepackage{fancyhdr}
\usepackage{datetime}

\usepackage[round]{natbib}

\newcommand{\singbr}[1]{\left[#1\right]}
\newcommand{\myset}[1]{\left\{#1\right\}}
\newcommand{\paren}[1]{\left(#1\right)}

\newcommand{\expected}[2]{\underset{#1}{\E}\singbr{#2}}

\newcommand{\abs}[1]{\left|#1\right|}

\newcommand{\wh}[1]{\ensuremath{\widehat{#1}}}

\DeclareFontFamily{U}{mathx}{\hyphenchar\font45}
\DeclareFontShape{U}{mathx}{m}{n}{<-> mathx10}{}
\DeclareSymbolFont{mathx}{U}{mathx}{m}{n}
\DeclareMathAccent{\widebar}{0}{mathx}{"73}

\newcommand{\ep}{\ensuremath{\epsilon}}
\newcommand{\de}{\ensuremath{\delta}}

\newcommand{\si}{\ensuremath{\sigma}}

\newcommand{\R}{\ensuremath{\mathbb{R}}}

\newcommand{\Zi}{\ensuremath{\mathbb{Z}}}

\newcommand{\E}{\ensuremath{\mathbb{E}}}

\newcommand{\ra}{\ensuremath{\rightarrow}}

\definecolor{light-gray}{gray}{0.80}

\definecolor{darkred}{rgb}{0.64, 0.0, 0.0}

\newtheorem{thm}{Theorem}[section]

\newtheorem{lem}[thm]{Lemma}

\newtheorem{defn}{Definition}[section]

\newenvironment{itemizesquish}{\begin{list}{\labelitemi}{\setlength{\itemsep}{-0.2em}\setlength{\labelwidth}{0.5em}\setlength{\leftmargin}{\labelwidth}
\addtolength{\leftmargin}{\labelsep}}}{\end{list}}

\usepackage[accepted]{aistats2020}
\makeatletter
\def\blfootnote{\gdef\@thefnmark{}\@footnotetext}
\makeatother

\begin{document}

\twocolumn[

\aistatstitle{Formal Limitations on the Measurement of Mutual Information}

\aistatsauthor{ David McAllester \And Karl Stratos }
\aistatsaddress{ Toyota Technological Institute at Chicago \And Rutgers University } ]

\begin{abstract}
  Measuring mutual information from finite data is difficult.
  Recent work has considered variational methods maximizing a lower bound.
  In this paper, we prove that serious statistical limitations are inherent to any method of measuring mutual information.
  More specifically, we show that any distribution-free high-confidence lower bound
  on mutual information estimated from $N$ samples cannot be larger than $O(\ln N)$.
\end{abstract}

\section{INTRODUCTION}

Mutual information has important applications to unsupervised learning.
It underlies classical representation learning methods such as Brown clustering \citep{brown1992class}, INFOMAX \citep{bell1995information},
and the information bottleneck method \citep{bottleneck}.
Maximizing mutual information is also central to recent works on unsupervised representation learning with neural networks \citep{IT-cotrain,MINE,Contrastive,PartOfSpeech,DIM}.

Unfortunately, measuring mutual information from finite data is a notoriously difficult estimation problem.
This difficulty has motivated researchers to consider more computationally amenable measurement methods that maximize a parameterized lower bound on mutual information.
The hope is that the optimized value of the bound is close to the true value of mutual information
and can serve as a good enough approximation.
For instance, this is the approach advocated by Mutual Information Neural Estimator (MINE) \citep{MINE}
and contrastive predictive coding (CPC) \citep{Contrastive}.

Here we prove that serious statistical limitations are inherent to all methods of measuring mutual information.
More specifically, we show that any distribution-free high-confidence lower bound on mutual information cannot be larger than $O(\ln N)$ where $N$ is the number of samples.
Thus a meaningful high-confidence lower bound is infeasible when underlying mutual information is large (e.g., hundreds of bits).

Our result corrects, generalizes, and unifies past work on measuring mutual information.
Unlike previous intractability results that are tied to specific estimators \citep{KNN-MI2,Contrastive},
our result is universal to all estimators.
Hence we show that without making strong assumptions on the population distribution, such as the small support assumption in \citet{valiant2011estimating} and minimax bounds \citep{jiao2015minimax},
it is generally impossible to guarantee an accurate estimate of mutual information.  Our results contradict a theorem
in \citet{MINE} claiming a polynomial sample size convergence rate for a mutual information estimator and we point out an error in their proof.

While it is infeasible to give meaningful high-confidence lower bound guarantees for large mutual information, estimators lacking formal guarantees might still be useful in practice.
To this end, we propose expressing mutual information as a difference of entropies and estimating the entropy terms by cross-entropy upper bounds.
This difference-of-entropies (DoE) estimator gives neither an upper bound nor a lower bound guarantee on mutual information.
Nevertheless, we give theoretical and empirical evidence that DoE can meaningfully estimate large mutual information from feasible samples.

\subsection{Overview}
\label{sec:overview}
We state our main result below. We write $(X, Y)$ to denote a pair of random variables with ranges $(\mathcal{X}, \mathcal{Y})$.
Unless otherwise specified, they can be discrete or continuous.
For simplicity, all distributions are assumed to have full support.

\begin{thm}
  Let $B$ be any mapping from $N$ samples of $(X,Y)$ to a real number with the following property:
  for any distribution $p_{XY}$ over $(X, Y)$,
  given $N$ iid samples $(x_1, y_1) \ldots (x_N, y_N) \sim p_{XY}$, with probability at least 0.99
  \begin{align*}
    I(X,Y; p_{XY}) \geq B((x_1, y_1) \ldots (x_N, y_N))
  \end{align*}
  where $I(X,Y; p_{XY})$ is the mutual information between $(X,Y)$ under $p_{XY}$.
  Now pick any distribution $q_{XY}$ over $(X,Y)$.
  Then given $N \geq 50$ iid samples $(x'_1, y'_1) \ldots (x'_N, y'_N) \sim q_{XY}$, with probability at least 0.96
  \begin{align*}
    B((x'_1, y'_1) \ldots (x'_N, y'_N)) \leq 2 \ln N + 5
  \end{align*}
  \label{thm:main}
\end{thm}

In the rest of the paper, we build toward this result by proving statistical limitations on
measuring lower bounds on Kullback-Leibler (KL) divergence and entropy.
We derive several intermediate results, specifically:

\begin{itemizesquish}
\item We first consider the Donsker-Varadhan lower bound on KL divergence which underlies the approach in MINE.
  We give an intuitive explanation on why estimating the bound from samples is problematic.
  We show that the polynomial sample complexity result given by \citet{MINE} is incorrect.
\item We formally prove that any lower bound on KL divergence cannot be larger than $O(\ln N)$. This subsumes the Donsker-Varadhan bound as a special case.
\item We formally prove that any lower bound on entropy cannot be larger than $O(\ln N)$. Theorem~\ref{thm:main} is a special case of this result.
\item We motivate measuring mutual information as a difference of entropies, each of which measured by minimizing cross entropy.
  We empirically show that it outperforms existing variational lower bounds in synthetic experiments
  and produce realistic estimates of mutual information in real-world datasets.
\end{itemizesquish}

\section{ISSUES WITH THE DONSKER-VARADHAN LOWER BOUND}
\label{sec:motivation}

The Donsker-Varadhan (DV) lower bound on KL divergence is stated below. For completeness, a simple proof is given in the supplementary material. 

\begin{thm}[\citealp{DV}]
  Let $p_X$ and $q_X$ be distributions over $X$ with finite KL divergence.
  Then for all bounded functions $f:\mathcal{X} \ra \R$
  \begin{align}
    D_{\mathrm{KL}}(p_X||q_X) \geq  \mathrm{DV}_f(p_X||q_X) \label{ineq:dv}
  \end{align}
  where
  \begin{align}
    \mathrm{DV}_f(p_X||q_X) := \expected{x \sim p_X}{f(x)} - \ln \expected{x \sim q_X}{e^{f(x)}} \label{def:dv}
  \end{align}
  Moreover, \eqref{ineq:dv} holds with equality for some $f$ with range $[0, F_{\mathrm{max}}]$.
  \label{thm:dv}
\end{thm}

The DV bound \eqref{def:dv} computes the difference between the expected value of $f(x)$ under $p_X$
and the log of the expected value of the exponential of $f(x)$ under $q_X$.
It can be easily estimated by sampling: given $x_1 \ldots x_N \sim p_X$ and $x'_1 \ldots x'_N \sim q_X$, we can compute the empirical estimate

{\small
\begin{align}
  \wh{\mathrm{DV}}_f^N(p_X||q_X) := \frac{1}{N} \sum_{i=1}^N  f(x_i) - \ln \paren{ \frac{1}{N} \sum_{i=1}^N e^{f(x'_i)} } \label{def:dv-emp}
\end{align}
}

\noindent
An application of the DV bound is measuring KL divergence.
Specifically, we estimate KL divergence by estimating the associated DV bound from samples:
\begin{align*}
  D_{\mathrm{KL}}(p_X||q_X) &\geq \mathrm{DV}_f(p_X||q_X) \\ &\gtrsim \wh{\mathrm{DV}}_f^N(p_X||q_X)
\end{align*}
The first inequality holds for all choices of $f$ by Theorem~\ref{thm:dv}.
We require that the second inequality holds with high probability (with respect to random sampling).
We now show that this approach to measuring KL divergence is problematic.

\subsection{Statistical Limitations on Measuring the DV Bound}
\label{subsec:dv}

A crucial observation is that the second term in the DV bound \eqref{def:dv} involves an expectation of the exponential
\begin{align*}
  \expected{x \sim q_X}{e^{f(x)}}
\end{align*}
This expression has the same form as the moment generating function used in analyzing large deviation probabilities.
The utility of expectations of exponentials in large deviation theory is that such expressions can be dominated by extremely rare events (large deviations).
The rare events dominating the expectation will never be observed by sampling from $q_X$.

To quantitatively analyze the risk of unseen outlier events, we will make use
of the following simple lemma where we write $p_X(\Phi[x])$
for the probability over drawing $x$ from $p_X$ that the statement $\Phi[x]$ holds.

\begin{lem}[Outlier risk lemma]
Given $N \geq 2$ samples from $p_X$ and a property $\Phi[x]$ such that $p_X(\Phi[x]) \leq 1/N$,
the probability that no sample $x$ satisfies $\Phi[x]$ is at least 1/4.
\label{lem:outlier}
\end{lem}

\begin{proof}
  The probability that $\Phi[x]$ is unseen in the sample is at least $(1-1/N)^N$ which is at least $1/4$ for $N \geq 2$
  and where $\lim_{N \rightarrow \infty} (1-1/N)^N = 1/e$.
\end{proof}

We can use the outlier risk lemma to perform a quantitative risk analysis of the DV bound.
Assume without loss of generality that $[0, F_{\max}]$ is the range of $f$ taken on $\mathcal{X}$.
Let us consider the best case scenario in which the empirical estimate \eqref{def:dv-emp} is the largest.
It is easy to see that the largest value is $F_{\max}$ and attained when
\begin{align*}
  f(x_i) &=F_{\max} &&\forall i = 1 \ldots N\\
  f(x'_i) &= 0    &&\forall i = 1 \ldots N
\end{align*}
where $x_1 \ldots x_N$ are samples from $p_X$ and $x'_1 \ldots x'_N$ are samples from $q_X$.
But by the outlier risk lemma, there is still at least a 1/4 probability that
\begin{align*}
  \expected{x \sim q_X}{e^{f(x)}} \geq \frac{1}{N} e^{F_{\max}}
\end{align*}
Since we require that \eqref{def:dv-emp} is a high-confidence lower bound on the DV bound, it must account for the unseen outlier
risk.\footnote{We intentionally keep this argument informal to give intuition since the result on the DV bound will be subsumed by the general result in Theorem~\ref{thm:kl}.}
In particular, we must have
\begin{align*}
  \wh{\mathrm{DV}}_f^N(p_X||q_X) & \leq F_{\max} - \ln\frac{e^{F_{\max}}}{N} =  \ln N
\end{align*}

\subsection{Discussion of MINE}
\label{sec:mine}

Our investigation of measuring the DV bound from samples is motivated by \citet{MINE} who propose measuring and maximizing mutual information by formulating it as KL divergence and considering the DV lower bound.
More specifically, they introduce a function $f: \mathcal{X} \times \mathcal{Y} \ra \R$ parameterized by a neural network and perform gradient descent to optimize
\begin{align}
  \sup_f \frac{1}{N} \sum_{i=1}^N f(x_i, y_i) - \ln \frac{1}{N} \sum_{i=1}^N e^{f(x'_i, y'_i)} \label{eq:mine}
\end{align}
where $(x_1, y_1) \ldots (x_N, y_N)$ are drawn from $p_{XY}$ and $(x'_1, y'_1) \ldots (x'_N, y'_N)$ are drawn from $p_X \times p_Y$.
This is an empirical estimate of the DV bound on
\begin{align}
  D_{\mathrm{KL}}(p_{XY}||p_X \times p_Y) = I(X,Y; p_{XY}) \label{eq:mikl}
\end{align}
They claim that \eqref{eq:mine} yields a high-confidence accurate measurement of underlying mutual information with polynomial sample complexity under mild assumptions (Theorem 3 in their paper),
which apparently contradicts our previous observation.

Upon inspection, we have found that their claim is wrong.
In the appendix of the arXiv version v4, they make an incorrect application of the
Hoeffding inequality in equation (46) from which they incorrectly derive equation (49).
Hoeffding depends on the bounded range of the random variable: (49) is bounding the exponential of the variable.
Thus their bound stated in Theorem 3 has an exponential dependence on the variable $M$.

\section{STATISTICAL LIMITATIONS ON MEASURING LOWER BOUNDS ON KL DIVERGENCE}

The analysis given in Section~\ref{subsec:dv} is specific to the DV bound.
One might hope that there exists a different lower bound on KL divergence
(e.g., $D_{\mathrm{KL}}(p_X||q_X) = \sup_{f > 0} \expected{x \sim p_X}{\ln f(x)} - \expected{x \sim q_X}{f(x)} + 1$
in \citet{nguyen2010estimating} which does not involve an expectation of the exponential)
that overcomes the limitation.

We now strengthen the result by formally proving that no lower bound on KL divergence estimated from samples can be large.
We consider a more challenging setting in which we have perfect knowledge of $p_X$ (i.e., we can compute probabilities under the distribution)
and only sample from $q_X$ to estimate a lower bound on $D_{\mathrm{KL}}(p_X||q_X)$.
Even in this setting, we have the following negative result.

\begin{thm}
  Let $B$ be any distribution-free high-confidence lower bound on $D_{\mathrm{KL}}(p_X||q_X)$ computed with complete knowledge of
  $p_X$ but only samples from $q_X$.

More specifically, let $B(p_X,S,\delta)$ be any real-valued function of a distribution $p_X$, a multiset $S$,
and a confidence parameter $\delta$  such that, for any $p_X$, $q_X$ and $\delta$,
with probability at least $1-\delta$ over a draw of $S$ from $q_X^N$ we have
\begin{align*}
  D_{\mathrm{KL}}(p_X||q_X) \geq B(p_X,S,\delta)
\end{align*}
For any such bound, and for $N \geq 2$, for any $q_X$, with probability at least $1 -4\delta$
over the draw of $S$ from $q_X^N$ we have
\begin{align*}
  B(p_X,S,\delta) \leq \ln N
\end{align*}
\label{thm:kl}
\end{thm}

\begin{proof}
  Consider distributions $p_X$ and $q_X$ and $N \geq 2$. Define $\tilde{q}_X$ by
  \begin{align*}
    \tilde{q}_X(x) = \left(1-\frac{1}{N}\right)q_X(x) + \frac{1}{N}p_X(x)
  \end{align*}
  We now have $D_{\mathrm{KL}}(p_X||\tilde{q}_X) \leq \ln N$.
  We will prove that from samples $S \sim q_X^N$ we cannot reliably distinguish between $q_X$ and $\tilde{q}_X$.

  The premise is that $B$ yields a high-confidence lower bound for any distribution. Thus we have
  \begin{align}
    \Pr_{S \sim \tilde{q}_X^N}(\mathrm{Small}(S)) \geq 1-\delta \label{eq:kl-key}
  \end{align}
  where $\mathrm{Small}(S)$ represent the event that $B(p_X,S,\delta) \leq \ln N$.
  The distribution $\tilde{q}_X$ equals the marginal on $x$ of
  a distribution on pairs $(s,x)$ where $s$ is the value of Bernoulli variable with bias $1/N$ such that if
  $s = 1$ then $x$ is drawn from $p_X$ and otherwise $x$ is drawn from $q_X$. By Lemma~\ref{lem:outlier}
  the probability that all coins are zero is at least 1/4.  Conditioned on all coins being zero the distributions
  $\tilde{q}_X^N$ and $q_X^N$ are the same. Let $\mathrm{Pure}(S)$ represent the event that all coins are 0.
  We now have
  \begin{align*}
    &\Pr_{S\sim q_X^N}(\mathrm{Small(S)})\\
    &= \Pr_{S \sim \tilde{q}_X^N}(\mathrm{Small}(S)|\mathrm{Pure}(S)) \\
    &= \frac{\Pr_{S \sim \tilde{q}_X^N}(\mathrm{Pure}(S) \wedge \mathrm{Small(S)})}{\Pr_{S \sim \tilde{q}_X^N}(\mathrm{Pure}(S))} \\
    &\geq \frac{\Pr_{S \sim \tilde{q}_X^N}(\mathrm{Pure}(S)) -\Pr_{S \sim \tilde{q}_X^N}(\neg\mathrm{Small}(S))} {\Pr_{S \sim \tilde{q}_X^N}(\mathrm{Pure}(S))} \\
    &\geq \frac{\Pr_{S \sim \tilde{q}_X^N}(\mathrm{Pure}(S)) - \delta}{\Pr_{S \sim \tilde{q}_X^N}(\mathrm{Pure}(S))} \\
    &= 1 - \frac{\delta} {\Pr_{S \sim \tilde{q}_X^N}(\mathrm{Pure}(S))} \\
    &\geq 1 - 4\delta
  \end{align*}

\end{proof}

Note the importance of the fact that the lower bound $B$ is distribution-free (the DV bound is distribution-free): this allows us to construct an adversarial distribution $\tilde{q}_X$
with small KL divergence and turn the premise on its head to upper bound $B$ \eqref{eq:kl-key}.
It may be possible to construct distribution-specific lower bounds that do not suffer the same limitations.
Theorem~\ref{thm:kl} proves that without making such additional assumptions, it is not possible to guarantee a large lower bound on KL divergence.

This result has immediate implications on measuring mutual information which is a special case of KL divergence \eqref{eq:mikl}.
A direct application of Theorem~\ref{thm:kl} gives the following:
in the setting in which we have perfect knowledge of $p_{XY}$ but can only sample from $p_X$ and $p_Y$ (i.e., we cannot compute marginals)
and measure a lower bound on $I(X,Y;p_{XY})$ from $N$ samples, we cannot guarantee that the bound is larger than $\ln N$.
However, this setting is arguably awkward.
In the next section, we make a more natural argument by considering entropy.

\section{STATISTICAL LIMITATIONS ON MEASURING LOWER BOUNDS ON ENTROPY}

Recall that mutual information can be formulated as a difference of entropies
\begin{align}
  I(X,Y; p_{XY}) = H(X; p_X) - H(X|Y; p_{XY}) \label{eq:ent-diff}
\end{align}
where $H(X; p_X)$ is the entropy of $X$ under $p_X$ and $H(X|Y; p_{XY})$ is the conditional entropy of $X$ given $Y$ under $p_{XY}$.
Entropy is nonnegative for discrete variables: in this case we have
\begin{align*}
  I(X,Y; p_{XY}) \leq H(X; p_X)
\end{align*}
It states that the mutual information between $X$ and $Y$ cannot be larger than information content of $X$ alone.
Thus a lower bound on mutual information implies a lower bound on entropy.
We will show that any distribution-free high-confidence lower bound on entropy requires a sample size exponential in the size of the bound.

The above argument seems problematic for the case of continuous densities as differential entropy can be negative.
However, for the continuous case we have
\begin{align}
  I(X,Y; p_{XY}) = \sup_{C,C'}\; I(C(X),C'(Y); p_{XY}) \label{eq:binning}
\end{align}
where $C$ and $C'$ range over all maps from the underlying continuous space to discrete sets (all binnings of the continuous space).
A proof is given in the supplementary material. 
Hence an $O(\ln N)$ upper bound on the measurement of mutual information for the discrete case applies to the continuous case as well.
We assume the discrete case in this section without loss of generality.

We use the following definition.
\begin{defn}
  The type of a multiset $S$, denoted ${\cal T}(S)$, is a function on positive integers such that
  \begin{align*}
    {\cal T}(S)(i) := \abs{\myset{x \in S: \sum_{x' \in S:\; x' = x} 1 = i}}
  \end{align*}
  That is, ${\cal T}(S)(i)$ is the number of elements of $S$ that occur $i$ times in $S$.
\end{defn}

The type ${\cal T}(S)$ contains all information relevant to estimating the actual probability of the items of a given count
and of estimating the entropy of the underlying distribution. The problem of
estimating distributions and entropies from sample types has been investigated by various authors \citep{GT,AlwaysGT,AlwaysGT2,Arora}.
Here we give the following negative result on lower bounding the entropy of a distribution by sampling.

\begin{thm}
Let $B$ be any distribution-free high-confidence lower bound on $H(X; p_X)$ computed from a type
${\cal T}(S)$ with $S \sim p_X^N$.

More specifically, let $B({\cal T},\delta)$ be any real-valued function of a type ${\cal T}$
and a confidence parameter $\delta$ such that for any $p_X$,
with probability at least $1-\delta$ over a draw of $S$ from $p_X^N$, we have
\begin{align*}
  H(X;p_X) \geq B({\cal T}(S),\delta)
\end{align*}
For any such bound, and for $N \geq 50$ and $k \geq 2$, for any $p_X$, with probability at least $1 -\delta - 1.01/k$
over the draw of $S$ from $p_X^N$ we have
\begin{align*}
  B({\cal T}(S),\delta) \leq \ln 2kN^2
\end{align*}
\label{thm:entropy}
\end{thm}

\begin{figure}[ht]
  \begin{center}
    \includegraphics[scale=0.68]{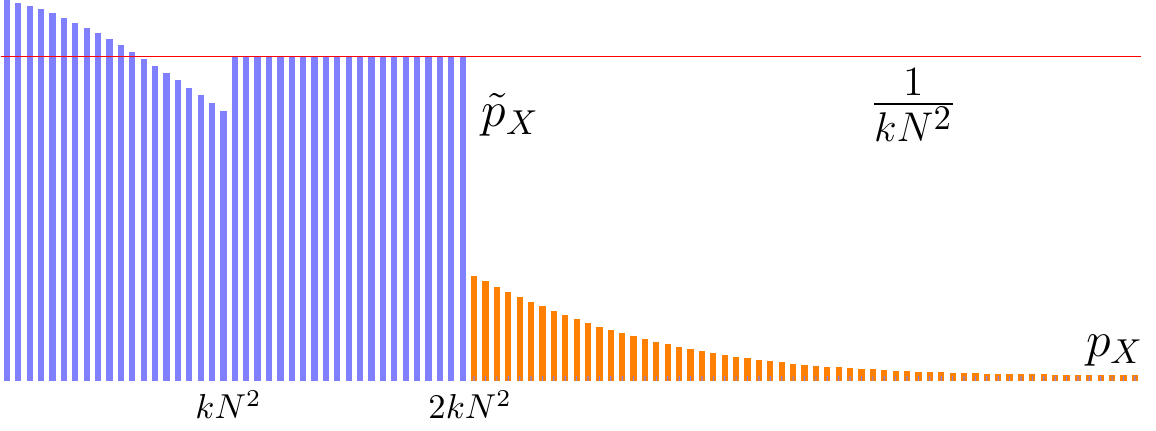}
  \end{center}
  \caption{Construction of an adversarial distribution $\tilde{p}_X$ from $p_X$ with entropy $H(X; \tilde{p}_X) \leq \ln2kN^2$.
  }
  \label{fig:entropy}
\end{figure}

\begin{proof}
Consider a distribution $p_X$ and $N \geq 100$. If the support of $p_X$ has fewer than $2kN^2$ elements then $H(X;p_X) < \ln 2kN^2$
and by the premise of the theorem
we have that, with probability at least $1-\delta$ over the draw of $S$, $B({\cal T}(S),\delta) \leq H(X;p_X)$ and the theorem follows.  If the support of $p_X$ has
at least $2kN^2$ elements then we sort the support of $p_X$ into a (possibly infinite) sequence $x_1,x_2,\ldots$ so that $p_X(x_i) \geq p_X(x_{i+1})$.
We then define a distribution $\tilde{p}_X$ on the elements $x_1 \ldots x_{2kN^2}$ by
\begin{align*}
\tilde{p}_X(x_i) =\left(\begin{array}{ll}
p_X(x_i) & \mbox{for $i \leq kN^2$} \\ \\
\frac{\mu}{kN^2} & \mbox{for $kN^2 < i \leq 2kN^2$}
\end{array}\right)
\end{align*}
where $\mu := \sum_{j > kN^2} p_X(x_j)$. See Figure~\ref{fig:entropy} for illustration.
We will let $\mathrm{Small}(S)$ denote the event that $B({\cal T}(S),\delta) \leq \ln2kN^2$
and let $\mathrm{Pure}(S)$ abbreviate the event that no element $x_i$ for $i > kN^2$ occurs twice in the sample.
Since $\tilde{p}_X$ has a support of size $2kN^2$ we have $H(X; \tilde{p}_X) \leq \ln2kN^2$.
Applying the premise of the lemma to $\tilde{p}_X$ gives
\begin{align}
  \Pr_{S \sim \tilde{p}_X^N}(\mathrm{Small}(S)) \geq 1 - \delta \label{step1}
\end{align}
For a type ${\cal T}$ let $\Pr_{S \sim P^N}({\cal T})$ denote the probability over drawing $S \sim P^N$
that ${\cal T}(S) = {\cal T}$.  We now have
\begin{align*}
  \Pr_{S \sim p_X^N}({\cal T}(S)|\mathrm{Pure}(S)) = \Pr_{S \sim \tilde{p}_X^N}({\cal T}(S)|\mathrm{Pure}(S))
\end{align*}
This gives the following.
\begin{align}
  & \Pr_{S \sim p_X^N}(\mathrm{Small}(S))  \notag \\
  & \geq \Pr_{S \sim p_X^N}(\mathrm{Pure}(S)\wedge \mathrm{Small}(S)) \notag \\
  & = \Pr_{S \sim p_X^N}(\mathrm{Pure}(S)) \;\Pr_{S \sim p_X^N} (\mathrm{Small}(S) \;|\;\mathrm{Pure}(S))  \notag \\
  & = \Pr_{S \sim p_X^N}(\mathrm{Pure}(S)) \;\Pr_{S \sim \tilde{p}_X^N} (\mathrm{Small}(S) \;|\;\mathrm{Pure}(S))  \notag \\
  &\geq \Pr_{S \sim p_X^N}(\mathrm{Pure}(S)) \;\Pr_{S \sim \tilde{p}_X^N} (\mathrm{Pure}(S) \wedge \mathrm{Small}(S)) \label{step2}
\end{align}
For $i > kN^2$ we have $\tilde{p}_X(x_i) \leq 1/(kN^2)$ which gives
\begin{align*}
  \Pr_{S\sim \tilde{p}_X^N}(\mathrm{Pure}(S)) \geq \prod_{j = 1}^{N-1}\;\left(1-\frac{j}{kN^2}\right)
\end{align*}
Using $1-z \geq e^{-1.01z}$ for $z \leq 1/100$ we have the following birthday paradox calculation.
\begin{align*}
  \ln \Pr_{S\sim \tilde{p}_X^N}(\mathrm{Pure}(S))
  &\geq  - \frac{1.01}{kN^2}\sum_{j=1}^{N-1}\; j \\
  &=  - \frac{1.01}{kN^2}\;\frac{(N-1)N}{2}  \\
  &\geq  - \frac{.505}{k}
\end{align*}
Therefore
\begin{align}
  \Pr_{S\sim \tilde{p}_X^N}(\mathrm{Pure}(S)) \geq e^{-.505/k} \geq 1- \frac{.505}{k}  \label{step3}
\end{align}
Applying the union bound to \eqref{step1} and \eqref{step3} gives
\begin{align}
  \Pr_{S \sim \tilde{p}_X^N}(\mathrm{Pure}(S) \wedge \mathrm{Small}(S)) \geq 1 - \delta - \frac{.505}{k}  \label{step4}
\end{align}
By a derivation similar to that of \eqref{step3} we get
\begin{align}
\Pr_{S\sim p_X^N}(\mathrm{Pure}(S)) \geq 1 - \frac{.505}{k}  \label{step5}
\end{align}
Combining \eqref{step2}, \eqref{step4} and \eqref{step5} gives
\begin{align*}
  \Pr_{S \sim p_X^N}(\mathrm{Small}(S)) \geq 1 - \delta - \frac{1.01}{k}
\end{align*}

\end{proof}

Again, we note the importance of the fact that the entropy lower bound $B$ is distribution-free:
this allows us to construct an adversarial distribution $\tilde{p}_X$ with small entropy and apply the premise to upper bound $B$ \eqref{step1}.
Theorem~\ref{thm:main} is derived from Theorem~\ref{thm:entropy} by choosing appropriate hyperparameter values $\de, k$
and using the fact that a lower bound on mutual information implies a lower bound on entropy.

\section{TOWARD ACCURATE MEASUREMENT OF MUTUAL INFORMATION}

Our main contribution is a set of fundamental statistical limitations on measuring a lower bound on mutual information
implied by the difficulty of measuring a lower bound on KL divergence (Theorem~\ref{thm:kl}) or entropy (Theorem~\ref{thm:entropy}).
But it is natural to ask: then how can we achieve an accurate measurement of mutual information?
As a complementary piece of contribution, in this section we explore a new estimator for mutual information that, while lacking formal guarantees, does not suffer from the same limitations.

\subsection{Mutual Information as a Difference of Entropies}

Since mutual information can be expressed as a difference of entropies \eqref{eq:ent-diff}, the problem
of measuring mutual information can be reduced to the problem of measuring entropies.
More specifically, we write mutual information as
\begin{align}
  &I(X,Y; p_{XY}) \notag \\
  &= \inf_{q_X}\; H(p_X, q_X) - \inf_{q_{X|Y}}\; H(p_{X|Y}, q_{X|Y}) \label{eq:proposal}
\end{align}
where we use the fact that the entropy $H(X; p_X)$ is upper bounded by the cross entropy between $p_X$ and $q_X$, denoted $H(p_X, q_X)$,
for all distributions $q_X$ (similarly for the conditional entropy $H(X|Y; p_{XY})$). They are equal iff $q_X = p_X$ and we have
\begin{align}
H(X; p_X) = \inf_{q_X}\; H(p_X, q_X) \label{eq:xent}
\end{align}
Note that the upper-bound guarantee for the cross-entropy estimator \eqref{eq:xent} yields neither an upper bound nor a lower bound guarantee for a difference of entropies \eqref{eq:proposal}.
However, we give theoretical evidence below that, unlike lower bound estimators, upper bound cross-entropy estimators can meaningfully estimate large entropies from feasible samples.

\paragraph{Cross entropy estimation.}
The statistical limitations on distribution-free high-confidence lower bounds on entropy do not arise for cross-entropy upper bounds.
For upper bounds we can show that naive sample estimates of the cross-entropy loss produce meaningful (large entropy) results.
The empirical cross-entropy loss computed on samples $x_1 \ldots x_N$ from a population distribution $p_X$ is
\begin{align}
  \wh{H}^N(p_X,q_X) = \frac{1}{N} \sum_{i=1}^N\; -\ln q_X(x_i) \label{eq:xent-emp}
\end{align}
where $q_X$ is viewed as a model of $p_X$.
We can bound the true loss of $q_X$ by ensuring a minimum probability $e^{-F_{\max}}$ where $F_{\max}$ is then the maximum possible log loss in the cross-entropy objective.\footnote{In language modeling a loss bound exists for any model that ultimately backs off to a uniform distribution on characters.}
Given a loss bound of $F_{\max}$, $\wh{H}^N(p_X,q_X)$ is just the standard sample mean estimator of an expectation of a bounded variable.
In this case we have the following standard confidence interval derived from the Chernoff bound.

\begin{thm}
For any population distribution $p_X$, and model distribution $q_X$ with $-\ln q_X(x)$ bounded to the interval $[0,F_{\max}]$,
with probability at least $1 -\delta$ over the draw of $x_1 \ldots x_N \sim p_X$ we have
\begin{align*}
  H(p_X,q_X) \in \wh{H}^N(p_X,q_X) \pm F_{\max}\sqrt{\frac{ \ln \frac{2}{\delta}}{2N}}
\end{align*}
\end{thm}

Thus unlike high-confidence distribution-free lower bounds,
high-confidence distribution-free upper bounds on entropy can approach the true cross entropy at the modest sample rate of $1/\sqrt{N}$ even when the true cross entropy
is large.\footnote{It is also possible to give PAC-Bayesian bounds on $H(p_X, q_X^\theta)$ as functions of the parameters $\theta$ of $q_X$.
  See the supplementary material for details.}

\begin{table*}[t!]
  \caption{Estimates of mutual information under different estimators.
    Each estimator is trained for 3,000 steps where at every step it receives $N=128$ samples of $(X,Y)$ and optimizes its weights;
    it is fully tuned over hyperparameter choices with respect to its final estimate.
    In each row, we boldface the estimate closest to the ground-truth mutual information.\\
  }
  \label{tab:syn-mi}
  \centering
    \begin{tabular}{|c|c|c|c|c|c|c|c||c|c|}
    \hline
    DV & MINE & NWJ & NWJ (JS) & CPC & CPC+NWJ & DoE (Gaussian) & DoE (Logistic) & $I(X,Y)$ & $\ln N$ \\
    \hline
    2.72 & 2.57 & 1.99 & 1.50 & 2.73 & 2.77 & 4.19 & \textbf{4.13} & 4.13 & 4.85 \\
    \hline
    10.27 & 9.38 & 9.25 & 5.55 & 4.82 & 8.18 & 18.38 & \textbf{18.42} & 18.41 & 4.85 \\
    \hline
    61.96 & 34.56 & 50.46 & 13.41 & 4.85 & 10.45 & \textbf{104.18} & 104.16 & 106.29 & 4.85 \\
    \hline
    \end{tabular}

\end{table*}

\begin{figure*}[t!]
  \begin{center}
    \includegraphics[scale=0.6]{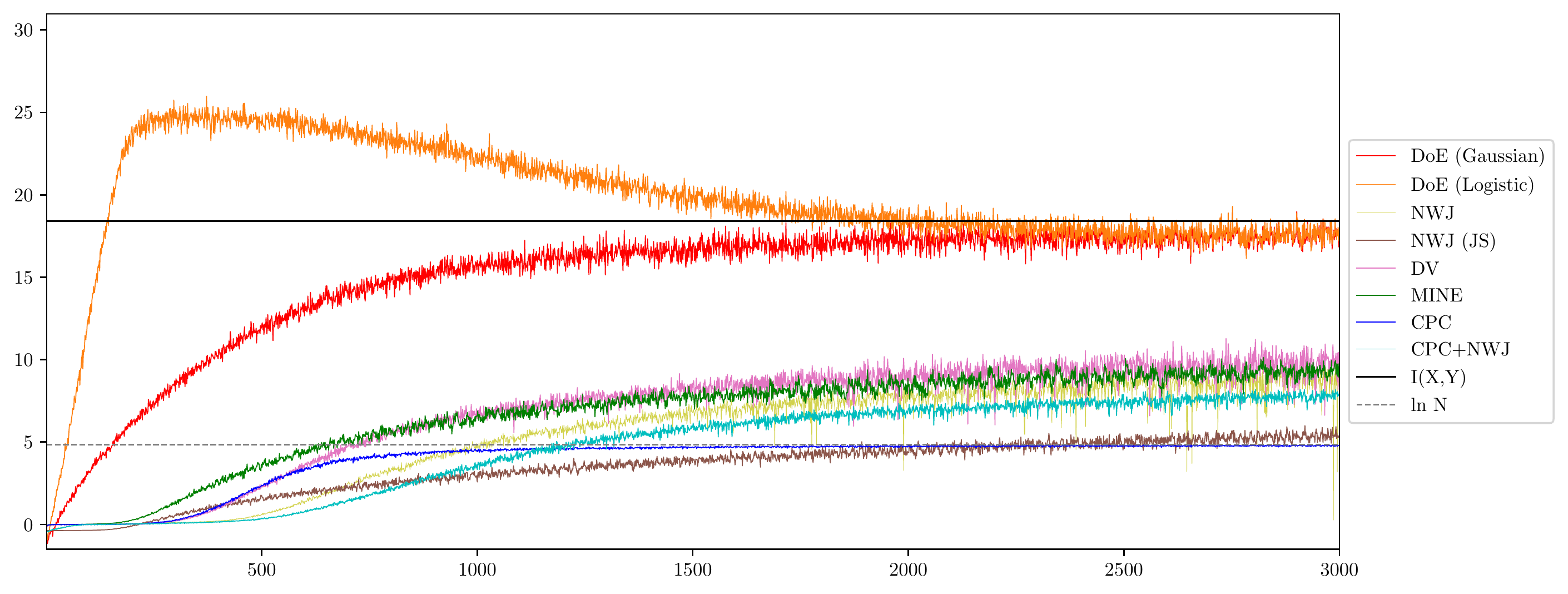}
  \end{center}
  \caption{A plot of mutual information estimates during training.
    At every training step ($x$-axis), each estimator receives a (shared) minibatch of $N=128$ samples,
    computes the current estimate ($y$-axis), and updates weights.
    We show the best-case scenario for each estimator by fully tuning its hyperparameters with respect to the final estimate.
    }
  \label{fig:synthetic-large}
\end{figure*}

\subsection{Experiments}

We present experiments with the proposed difference-of-entropies (DoE) estimator to gain a better understanding of its empirical behavior.
First, we compare DoE with existing lower-bound estimators in a standard synthetic setting based on correlated Gaussians.
Next, we apply DoE on the task of measuring mutual information between related articles and translation pairs and show an evidence of large mutual information.

In the following, DoE computes an empirical estimate of \eqref{eq:proposal} by taking
iid samples $(x_1, y_1) \ldots (x_N, y_N)$ from $p_{XY}$ and computing
\begin{align}
  \wh{H}(p_X, q_X) &= \inf_{q_X}\; \frac{1}{N} \sum_{i=1}^N - \log q_X(x_i) \notag\\
  \wh{H}(p_{X|Y}, q_{X|Y}) &= \inf_{q_{X|Y}}\; \frac{1}{N} \sum_{i=1}^N - \log q_{X|Y}(x_i|y_i) \notag\\
  \wh{I}(X,Y; p_{XY}) &= \wh{H}(p_X, q_X) - \wh{H}(p_{X|Y}, q_{X|Y}) \label{eq:mi-emp}
\end{align}
where each minimization corresponds to fitting a probabilistic model on the samples with the cross-entropy loss.

\subsubsection{Synthetic Experiments}

Following \citet{poole2019variational}, we define
random variables $X, Y \in \R^d$ where $(X_i, Y_i)$ are standard normal with correlation $\rho \in [-1, 1]$.
It can be checked that $I(X, Y) = - (d/2) \ln (1 - \rho^2)$.
We use $d=128$ and vary $\rho$ to experiment with different values of mutual information.
We compare with the following lower-bound estimators:
the DV bound \eqref{def:dv},
MINE (i.e., DV with an ``improved gradient estimator'') \citep{MINE},
NWJ \citep{nguyen2010estimating},
NWJ estimated by optimizing Jensen-Shannon divergence (NWJ (JS)),
CPC \citep{Contrastive},
and a nonlinear interpolation between NWJ and CPC (NWJ+CPC).
We refer the reader to \citet{poole2019variational} for a detailed exposition of these estimators.
We parameterize the distributions in DoE (i.e., $q_X$ and $q_{X|Y}$ in \eqref{eq:mi-emp}) by isotropic Gaussian (correct) or logistic (misspecified).

Table~\ref{tab:syn-mi} shows mutual information estimates given by these estimators.
All estimators are trained for 3,000 steps where at every step they use $N=128$ samples drawn from $p_{XY}$ to update their weights.
We tune the hyperparameters of each estimator (e.g., learning rate, number of hidden units, initialization strategies,
estimator-specific hyperparameters such as the mixing weights in NWJ+CPC and MINE) to minimize $\abs{I(X,Y) - \hat{m}}$ where $\hat{m}$ is the final estimate.
Thus the results assume an oracle that gives an optimal configuration.
All details can be found in the code: \url{https://github.com/karlstratos/doe}.

It is clear that (1) DoE obtains the most accurate estimates whether $I(X,Y) > \ln N$ or $I(X,Y) \leq \ln N$,
(2) DoE is the only estimator that achieves accurate estimates when underlying mutual information is large, and
(3) either Gaussian or logistic parameterization of DoE yields accurate estimates.
Furthermore, Table~\ref{tab:syn-mi} does not show the fact that DV, MINE, and NWJ are highly unstable (especially in the large mutual information setting).
In particular, while it seems that they can achieve estimates much larger than $\ln N$,
it is not representative of their general performance which is fraught with numerical overflow/underflow problems leading to completely off estimates.
\citet{poole2019variational} discuss the large variance of these estimators.
In contrast, DoE is based on the standard cross-entropy loss and very stable.

Figure~\ref{fig:synthetic-large} shows a plot of the best training session for each estimator.
DoE is a clear outlier as the only accurate estimator of mutual information.
Unlike lower-bound estimators, DoE can approach $I(X,Y)$ from above or below.

\subsubsection{Mutual Information Between Articles and Translations}

We consider two datasets for the choice of $(X,Y)$:
\begin{enumerate}
\item Related news article pairs extracted from the Who-Did-What dataset \citep{onishi2016did}
\item English-German translation pairs extracted from the IWSLT 2014 dataset
\end{enumerate}
We expect that mutual information is large in either setting:
a random article (sentence) has large entropy, but given a related article (translation) the uncertainty is drastically reduced.
We use a standard LSTM language model for $q_X$ and a standard attention-based translation model for $q_{X|Y}$.
More details of the experiments can be found in the supplementary material.

Table~\ref{tab:mi} shows the estimates of mutual information on the test portion of data.
We use log base 2 to accommodate the bit interpretation of entropies (rather than nats).
We see that mutual information is estimated to be over 120 bits on related article pairs and 54 bits on translation pairs.
Mutual information is estimated to be close to zero for shuffled pairs: this shows that the estimator can also handle small mutual information.

\begin{table}[t!]
  \caption{Estimates of mutual information (in bits) on article pairs and translation pairs based on the difference of entropies~\eqref{eq:mi-emp}.\\}
  \label{tab:mi}
  \centering
    \begin{tabular}{l|r}
    \hline
    distribution $p_{XY}$       &  \eqref{eq:mi-emp} \\
    \hline
    related article pairs      &  $120.34$   \\
    shuffled article pairs     &  $-2.38$  \\  
    translation pairs          &  $54.72$ \\
    shuffled translation pairs &  $-2.64$ \\  
    \hline
    \end{tabular}

\end{table}

\section{RELATED WORK}

We make a few additional remarks on related work to better contextualize our work.
In the continuous setting, a classical approach to estimating mutual information is based on computing the average
log of the distance to the $k$-th nearest neighbor in samples \citep{KNN-MI}.
\citet{KNN-MI2} show that this estimator suffers exponential sample complexity and propose more refined nearest-neighbor methods \citep{KNN-MI2}.
In contrast, we establish that serious statistical limitations are inherent to the measurement of mutual information no matter what estimator is used.

There is a line of work that develops efficient estimators for entropy by assuming distributions with very small support---distributions with support smaller than the sample size.
\citet{valiant2011estimating} show that it is possible to achieve an optimal sample rate of $O(n/\ln n)$ where $n$ is the support size.
Past work on analyzing minimax bounds likewise assume small support \citep{jiao2015minimax,han2015minimax,kandasamy2015nonparametric}.
In this case, the entropy of the distribution cannot be larger than the log of the number of samples.
This is in agreement with, but does not imply, our results.
We are interested in the large entropy setting, such as a distribution over all possible images or articles.
We cannot have the number of samples equal to the number of possible images or articles.

As discussed in depth in Section~\ref{sec:motivation}, we are motivated by the approach in MINE which measures the DV bound to estimate and maximize mutual information \citep{MINE}.
CPC is another notable example that illustrates the statistical limitations of measuring mutual information \citep{Contrastive}.
CPC maximizes a lower bound on mutual information through noise contrastive estimation:
it is shown that this lower bound cannot be larger than $\ln k$ where $k$ is number of negative samples used in the contrastive choice.
Complementary to our work, a recent work by \citet{poole2019variational} investigates tradeoffs between bias and variance in estimating variational bounds on mutual information.

There is a class of representation learning methods such as Brown clustering \citep{brown1992class} and the information bottleneck method \citep{bottleneck} that maximize a lower bound on mutual information given by the data processing inequality (DPI).
In these methods, we learn ``coding'' functions $(C, C')$ by optimizing the objective
\begin{align*}
  \max_{C, C'} I(C(X), C'(Y); p_{XY}) \leq I(X,Y; p_{XY})
\end{align*}
where the inequality is by the DPI.
Information theoretic co-training \citep{IT-cotrain} considers a similar lower bound and has been shown to be useful for label induction in speech and text \citep{PartOfSpeech}.
Measuring these lower bounds is subject to the same limitations presented in this paper.

\section{CONCLUSIONS}

Maximizing mutual information is well motivated as a method of unsupervised pretraining of representations that maintain semantic signal while dropping uninformative noise.
However, measuring and maximizing mutual information from finite data is a difficult training objective.
In this paper, we have shown serious statistical limitations inherent to measuring lower bounds on various information theoretic measures including KL divergence, entropy, and mutual information.
We have also given theoretical arguments that representing mutual information as a difference of entropies,
and estimating those entropies by minimizing cross-entropy loss, is a more statistically justified approach than maximizing a lower bound on mutual information.

Unfortunately cross-entropy upper bounds on entropy fail to provide
either upper or lower bounds on mutual information---mutual information is a difference of entropies.
We cannot rule out the possible existence of superintelligent models,
models beyond current expressive power, that dramatically reduce cross-entropy loss.
Lower bounds on entropy can be viewed as proofs of the non-existence of superintelligence.
We should not surprised that such proofs are infeasible.



\bibliography{mmi_limit}
\bibliographystyle{natbib}

\appendix

\section{PROOF OF THEOREM~2.1}
\label{app:dv}

For any distribution $r_X$ over $X$, we can write

{\small
\begin{align}
  D_{\mathrm{KL}}(p_X||q_X) &= \expected{x \sim p_X}{ \ln \frac{r_X(x)}{q_X(x)}} + D_{\mathrm{KL}}(p_X||r_X) \notag\\
  &\geq \expected{x \sim p_X}{ \ln \frac{r_X(x)}{q_X(x)}} \label{eq:dv-lb}
\end{align}
}

\noindent
Let $f:\mathcal{X} \ra \R$ be a bounded function and define
\begin{align*}
  r_X(x) = \frac{q_X(x) e^{f(x)}}{ \expected{x \sim q_X}{e^{f(x)}} }  &&\forall x \in \mathcal{X}
\end{align*}
which is a valid distribution over $X$.
Plugging this into the lower bound in \eqref{eq:dv-lb}, we have

{\small
\begin{align}
\expected{x \sim p_X}{ \ln \frac{r_X(x)}{q_X(x)}} = \expected{x \sim p_X}{f(x)} - \ln \expected{x \sim q_X}{e^{f(x)}} \label{eq:dv}
\end{align}
}

\noindent
By \eqref{eq:dv-lb}, the supremum of \eqref{eq:dv} over the choice of $f$ is precisely the KL divergence between $p_X$ and $q_X$.
It can be easily verified that an optimal $f$ is given by
\begin{align*}
f(x) = \ln \frac{p_X(x)}{q_X(x)}  &&\forall x \in \mathcal{X}
\end{align*}
Since \eqref{eq:dv} is invariant to translation of $f$, without loss of generality
we can assume that the range of $f$ is bounded in $[0, F_{\mathrm{max}}]$ for some constant $F_{\mathrm{max}}$.

\section{MUTUAL INFORMATION AS THE SUPREMUM OVER BINNINGS}
\label{app:binning}

We now show that the mutual information $I(X,Y; p_{XY})$ for $X$ and $Y$ continuous can be expressed as the supremum of $I(C(X),C'(Y); p_{XY})$ over discrete binnings of the continuous space.
We first consider the case where $X, Y \in \R$ and where the mutual information can be written as a Riemann integral over densities.
\begin{align*}
  &I(X,Y; p_{XY})\\ 
  &= \int p_{XY}(x,y)  \ln \frac{p_{XY}(x,y)}{p_X(x)p_Y(y)} \;dx\;dy \\
  &= \lim_{\ep \ra 0}\;\sum_{i,j \in \Zi}\; p_{XY}(i\ep,j\ep)\;\ln \frac{p_{XY}(i\ep,j\ep)}{p_X(i\ep)p_Y(j\ep)} \ep^2
\end{align*}
\noindent
where $\Zi$ is the set of all integers. For each $i \in \Zi$, define the half-open interval $C_{i, \ep} := [i \ep, (i+1) \ep)$.
  The probability of the interval is approximately $\ep p_X(i\ep)$ under $p_X$ (similarly for $p_Y$ and $p_{XY}$).
  Therefore we can write the last expression as
  \begin{align*}
    &\lim_{\ep \ra 0}\;\sum_{i,j \in \Zi}\; p_{XY}(C_{i,\ep}\times C_{j,\ep})\;\ln \frac{p_{XY}(C_{i,\ep}\times C_{j,\ep})}{p_X(C_{i,\ep})p_Y(C_{j,\ep})} \\
    &= \lim_{\ep \ra 0}\;\sum_{i,j \in \Zi}\; p_{I_\ep J_\ep}(i,j)\; \ln \frac{p_{I_\ep J_\ep}(i,j)}{p_{I_\ep}(i)p_{J_\ep}(j)} \\
    &= \lim_{\ep \ra 0}\;I(I_\ep,J_\ep; p_{I_\ep J_\ep})
  \end{align*}
  where $(I_\ep, J_\ep)$ denote the indices $(i, j)$ such that $x \in C_{i, \ep}$ and $y \in C_{j, \ep}$ for $(x,y) \sim p_{XY}$.

  This proof immediately generalizes to higher dimensions where the mutual information can be expressed as a Riemann integral.
  We believe that this statement remains true for arbitrary measures on product spaces where the mutual information is finite.
  However the proof for this extremely general case appears to be nontrivial.

\section{PAC-BAYESIAN BOUNDS}
\label{app:pac-bayes}

The PAC-Bayesian bounds apply to ``broad basin'' losses and loss estimates such as the following:
\begin{align*}
  H_\si(S,q_X^\theta)       &=  \expected{x \sim p_X}{ \expected{\epsilon \sim N(0,\si I)}{ -\ln q_X^{\theta+ \ep}(x) } } \\
  \wh{H}_\si(S,q_X^\theta) &=  \frac{1}{|S|} \sum_{x \in S}\; \expected{\epsilon \sim N(0,\si I)}{ -\ln q_X^{\theta + \ep}(x)}
\end{align*}
Under mild smoothness conditions on $q_X^\theta(x)$ as a function of $\theta$ we have
\begin{align*}
\lim_{\si \rightarrow 0} \; H_\si(p_X,q_X^\theta)    &= H(p_X,q_X^\theta) \\
\lim_{\si \rightarrow 0} \;\wh{H}_\si(S,q_X^\theta) &=  \wh{H}(S,q_X^\theta)
\end{align*}

\noindent
An $L_2$ PAC-Bayesian generalization bound \citep{PAC-Bayes} gives that for
any parameterized class of models and any bounded notion of loss, and
any $\lambda > 1/2$ and $\si > 0$,
with probability at least $1 - \delta$ over the draw of $S$ from $p_X^N$ we have the following simultaneously for all parameter vectors $\theta$.

{\small
\begin{align*}
& H_\si(p_X,q_X^\theta) \\
&\leq  \frac{1}{1 - \frac{1}{2\lambda}}\left(\wh{H}_\si(S,q_X^\theta) + \frac{\lambda F_{\max}}{N}\left(\frac{||\theta||^2}{2\si^2} + \ln \frac{1}{\delta}\right)\right)
\end{align*}
}

\noindent
It is instructive to set $\lambda = 5$ in which case the bound becomes.

{\small
  \begin{align*}
    &H_\si(p_X,q_X^\theta) \\
    &\leq \frac{10}{9}\left(\wh{H}_\si(S,q_X^\theta)+ \frac{5 F_{\max}}{N}\left(\frac{||\theta||^2}{2\si^2} + \ln \frac{1}{\delta}\right)\right)
  \end{align*}
}

\noindent
While this bound is linear in $1/N$, and tighter in practice than square root bounds, note that there is a small residual gap when holding $\lambda$ fixed at 5
while taking $N \rightarrow \infty$.
In practice the regularization parameter $\lambda$ can be tuned on holdout data.
One point worth noting is the form of the dependence of the regularization coefficient on $F_{\max}$, $N$ and the basin parameter $\si$.

It is also worth noting that the bound can be given in terms of ``distance traveled''
in parameter space from an initial (random) parameter setting $\theta_0$.

{\small
\begin{align*}
&H_\si(p_X,q_X^\theta) \\
& \leq   \frac{10}{9}\left(\wh{H}_\si(S,q_X^\theta) + \frac{5 F_{\max}}{N}\left(\frac{||\theta-\theta_0||^2}{2\si^2} + \ln \frac{1}{\delta}\right)\right)
\end{align*}
}

\noindent
Evidence is presented in \citet{NonVacuous} that the distance traveled bounds are tighter in practice than traditional $L_2$ generalization bounds.

\section{EXPERIMENT DETAILS}
\label{app:experiments}

\paragraph{Article pairs.}
We take pairs from the Who-Did-What dataset \citep{onishi2016did}.
The pairs in this dataset were constructed by drawing articles from the LDC Gigaword newswire corpus.
A first article is drawn at random and then a list of candidate second articles is drawn using the first sentence of the first article as an information retrieval query.
A second article is selected from the candidates using criteria described in \citet{onishi2016did},
the most significant of which is that the second article must have occurred within a two week time interval of the first.
The training statistics of this dataset after preprocessing is given in Table~\ref{tab:article-pair}.

\paragraph{Translation pairs.}
Our translation pairs consists of English-German sentence pairs extracted from the IWSLT 2014 dataset.
The training statistics of this dataset after preprocessing is given in Table~\ref{tab:translation}.

\begin{table}[t!]
  \centering
    \begin{tabular}{l|rr}
    \hline
                & train (tgt) & train (src) \\
                \hline
    \# articles  & 68348	& 68348 \\
    vocab size	& 100001	& 87941  \\
\# words	    &20271664	&19072167 \\
avg length	& 296&	279 \\
max length &	400&	400 \\
min length &	10&	12 \\
\hline
    \end{tabular}
      \caption{Training statistics of the article pairs}
  \label{tab:article-pair}
\end{table}

\begin{table}
  \centering
    \begin{tabular}{l|rr}
    \hline
                & train (tgt) & train (src) \\
                \hline
\# sentences&	160239&	160239 \\
vocab size&	24726&	35445 \\
\# words&	3275729&	3100720 \\
avg length &	20&	19 \\
max length&	175&	172 \\
min length&	2&	2 \\
\hline
    \end{tabular}
      \caption{Training statistics of the translation pairs}
  \label{tab:translation}
\end{table}

\paragraph{Model.}

We train an LSTM encoder-decoder model where the decoder doubles as both the decoder of a translation model and a language model.
The decoder is a left-to-right 2-layer LSTM in which a single word embedding matrix is used for both input embeddings and the softmax predictions.
When this model is trained as a language model on PTB using standard hyperparameter values it achieves test perplexity of 72.26.
The encoder is a separate left-to-right 2-layer LSTM using the same word embeddings as the decoder.
We use the input-feeding attention archietecture of \citet{luong2015effective}.

The model is trained using SGD and batch size 10 with no BPTT-style truncation.
The dimension of the input/hidden states is 900 (thus 1800 for the input-feeding decoder).
We use step-wise dropout with rate 0.65 on word embeddings and hidden states.
The model is trained for 40 epochs and the model that achieves the best validation perplexity is selected.
The sequence-level cross entropy is estimated as $\mathrm{SQXENT} = \frac{1}{M} \mathrm{NLL}$ where $\mathrm{NLL}$ is the negative log likelihood of the corpus and $M$ is the total number of sequences in the corpus.

Mutual information is estimated by taking the difference in $\mathrm{SQXENT}$ between the language model and the translation model (17).
For article pairs, we obtain
\begin{align*}
  \wh{I}(X,Y; p_{XY}) = 1131.74 - 1048.33 = 83.41
\end{align*}
in nats which translates to $120.34$ bits.
For translation pairs, we obtain
\begin{align*}
  \wh{I}(X,Y; p_{XY}) =  81.73 - 43.80 = 37.9
\end{align*}
in nats which translates to $54.72$ bits.

\end{document}